\documentclass[a4paper, 11pt]{article}

\usepackage{microtype}

\usepackage{etex,etoolbox}

\usepackage{amsmath}
\usepackage{amsfonts}
\usepackage{amssymb}
\usepackage{graphicx}
\usepackage{color}
\usepackage{tikz}
\usepackage{mathtools}
\usepackage{amsthm}

\usepackage{fullpage}
\theoremstyle{plain}
\newtheorem{theorem}{Theorem}
\newtheorem{lemma}[theorem]{Lemma}
\newtheorem{corollary}[theorem]{Corollary}

\theoremstyle{definition}

\definecolor{darkred}{rgb}{0.5,0,0}
\definecolor{darkblue}{rgb}{0,0,0.5}
\definecolor{darkgreen}{rgb}{0,0.5,0}

\usepackage{algorithm}
\usepackage{algorithmic}
\usepackage{setspace}
\newcommand{\ENSUREGAP}{\vspace{0.2cm}}
\definecolor{gray}{gray}{0.3}

\DeclareMathOperator{\id}{id}
\DeclareMathOperator{\str}{S}
\DeclareMathOperator{\supp}{supp}
\DeclareMathOperator{\Aut}{Aut}

\newcommand{\LanguageTYPESET}[1]{\mathsf{#1}}

\newcommand{\GIlang}{\LanguageTYPESET{GI}}
\newcommand{\GAlang}{\LanguageTYPESET{GA}}
\newcommand{\AUTlang}{\LanguageTYPESET{AUT}}

\newcommand{\cGIlang}{\LanguageTYPESET{col\text{-}GI}}
\newcommand{\cGAlang}{\LanguageTYPESET{col\text{-}GA}}
\newcommand{\cAUTlang}{\LanguageTYPESET{col\text{-}AUT}}

\newcommand{\Tourlang}{\LanguageTYPESET{Tour}}

\newcommand{\AsymTourlang}{\LanguageTYPESET{AsymTour}}
\newcommand{\Asymlang}{\LanguageTYPESET{Asym}}

\newcounter{claimcounter}
\newenvironment{claim}[1][]{
  \renewcommand{\proof}{\smallskip\par\noindent\textit{Proof of the claim. }}
  \medskip\par\noindent%
  \ifthenelse{\equal{#1}{}}{%
    \setcounter{claimcounter}{0}\refstepcounter{claimcounter}\textit{Claim~\arabic{claimcounter}.}
  }{%
\ifthenelse{\equal{#1}{resume}}{%
\refstepcounter{claimcounter}\textit{Claim~\arabic{claimcounter}.}%
}{%
\textit{Claim~#1.}%
}%
}%
}{%
\par\medskip%
}

\newcommand{\uend}{\hfill$\lrcorner$}

\title{A polynomial-time randomized reduction from tournament isomorphism to tournament asymmetry}

\author{Pascal Schweitzer \\ RWTH Aachen University\\
{\tt schweitzer@informatik.rwth-aachen.de} 
}

\usepackage[colorlinks]{hyperref}
\hypersetup{colorlinks
  ,linkcolor=darkred
  ,filecolor=darkgreen
  ,urlcolor=darkred
  ,citecolor=darkblue,
  linktocpage}

\begin{document}

\maketitle

\begin{abstract}

The paper develops a new technique to extract a characteristic subset from a random source that repeatedly samples from a set of elements. Here a characteristic subset is a set that when containing an element contains all elements that have the same probability.

With this technique at hand the paper looks at the special case of the tournament isomorphism problem that stands in the way towards a polynomial-time algorithm for the graph isomorphism problem. Noting that there is a reduction from the automorphism (asymmetry) problem to the isomorphism problem, a reduction in the other direction is nevertheless not known and remains a thorny open problem.

Applying the new technique, we develop a randomized polynomial-time Turing-reduction from the tournament isomorphism problem to the tournament automorphism problem. This is the first such reduction for any kind of combinatorial object not known to have a polynomial-time solvable isomorphism problem.
\end{abstract}

\section{Introduction.}\label{sec:introduction}

The graph automorphism problem asks whether a given input graph has a non-trivial automorphism. In other words the task is to decide whether a given graph is asymmetric. This computational problem is typically seen in the context of the graph isomorphism problem, which is itself equivalent under polynomial-time Turing reductions to the problem of computing a generating set for all automorphisms of a graph~\cite{DBLP:journals/ipl/Mathon79}. As a special case of the latter, the graph automorphism problem obviously reduces to the graph isomorphism problem. However, no reduction from the graph isomorphism to the graph automorphism problem is known.
In fact, while many computational problems surrounding structural equivalence of combinatorial objects can all be Turing-reduced to one another, the relationship between the graph automorphism and the graph isomorphism problem remains a repeatedly posed open question (see for example \cite{DBLP:journals/iandc/AgrawalA96, DBLP:journals/corr/AllenderGM15, Ghosh2014,MR1232421}).

With Babai's new ground-breaking algorithm~\cite{DBLP:conf/stoc/Babai16} that solves the graph isomorphism problem and thereby also the graph automorphism problem in quasi-polynomial time, the question arises whether it is possible to go further and devise a polynomial-time algorithm. For such an endeavor to succeed, special cases such as the group isomorphism and the tournament isomorphism problem, for which the currently fastest algorithms have a running time of~$n^{O(\log n)}$, should also be solvable in polynomial time. Tournaments, which are graphs in which between every pair of vertices there exists exactly one directed edge, also have an automorphism problem associated with them, asking whether a given tournament is asymmetric\footnote{Many publications in the context of graph isomorphism use the term rigid graph. However, the literature is inconsistent on the notion of a rigid graph, which can for example refer to having no non-trivial automorphism or no non-trivial endomorphism. We will use the notion asymmetric, which only ever means the former. Furthermore, we suggest the name graph asymmetry problem over graph automorphism problem,  so as not to confuse it with the computational problem to compute the automorphism group.}.
Again, for this problem the currently best running time is~$n^{O(\log n)}$ and analogously to general graphs there is a simple reduction from the automorphism problem to the isomorphism problem, but no reverse reduction has been known.

In this paper we show that there is a randomized polynomial-time Turing reduction from the tournament isomorphism problem to the tournament automorphism problem. This is the first such reduction for any kind of combinatorial object (apart from polynomial-time solvable cases of course).

The main new technical tool that we develop in the first part of the paper is a technique to exploit an oracle to the graph automorphism problem in order to obtain a non-trivial automorphism-invariant partition of a graph that is finer than the orbit partition (Sections~\ref{sec:sampling:subsets}--\ref{sec:sampling:minimimal:orbits}). We call the parts of such a partition suborbits. This technique is essentially applicable to all graph classes, not just tournaments. It hinges on a method to extract a characteristic subset from a random source that repeatedly samples from a set of elements. Here we say that a set is characteristic if it is a union of level sets of the probability function.

In the second part of the paper we show that, for tournaments, access to suborbits suffices to compute automorphism groups (Section~\ref{sec:auto:group:from:suborbits}). For this we adapt the group-theoretic divide and conquer approach of Luks~\cite{DBLP:journals/jcss/Luks82} to our situation. In this second part we exploit that the automorphism group of tournaments is solvable and we leave it as an open question whether something similar can be forged that is applicable to the group isomorphism problem (see Section~\ref{sec:open:prob}).

It might be worth noting that the techniques actually do not use any of the new structural insights from the quasi-polynomial-time algorithm of~\cite{DBLP:conf/stoc/Babai16}. Rather, the randomized sampling idea is heavily based on an older practical randomized algorithm designed to quickly detect non-isomorphism~(\cite{DBLP:conf/alenex/KutzS07,SchweitzerThesis}). It appears to be one of the few cases where randomization helps to derive a theoretical result for an isomorphism problem. We also borrow some ideas from a paper of Arvind, Das, and Mukhopadhyay concerned with tournament canonization~\cite{DBLP:journals/jcss/ArvindDM10}.

The necessity for randomization to obtain theoretical results in the context of isomorphism checking appears to be quite rare. The earliest result exploiting randomization seems to go to back to Babai~\cite{BabaiRandom} and is a randomized algorithm for checking isomorphism of graphs of bounded color class size. However that algorithm is actually a Las Vegas algorithm (an algorithm that does not make errors), and in the meantime deterministic algorithms are available~\cite{DBLP:conf/focs/FurstHL80}. However, for the new reduction in this paper it seems unclear how to remove the use of randomization and even how to remove the possibility for errors.

\subsection{Related work:} With respect to related work, we focus on results concerning graph automorphism as well as results concerning tournaments and refer the reader to other texts (for example~\cite{MR1373683,DBLP:conf/stoc/Babai16,MR1232421, DBLP:journals/jsc/McKayP14,  DBLP:conf/stacs/Schweitzer15}) for a general introduction to the graph isomorphism problem, current algorithms and overviews over complexity theoretic results.

\emph{(Tournament automorphism)} Let us start by highlighting two results specifically concerned with the tournament automorphism problem. Arvind, Das, and Mukhopadhyay~\cite{DBLP:journals/jcss/ArvindDM10} show that if tournament isomorphism is polynomial-time solvable then tournament canonization can be reduced in polynomial time to canonization of asymmetric tournaments. This implies now, with the result of the current paper, that from a canonization algorithm for asymmetric tournaments we can obtain a randomized canonization algorithm for tournaments in general. (In other words, the main theorem of our paper transfers to canonization.)
On the hardness side, Wager~\cite{DBLP:conf/mfcs/Wagner07,WagnerThesis} shows that tournament automorphism is hard for various circuit complexity classes ($\LanguageTYPESET{NL}$, $\LanguageTYPESET{C_=L}$, $\LanguageTYPESET{PL}$, $\LanguageTYPESET{DET}$, $\LanguageTYPESET{MOD_kL}$) under $\LanguageTYPESET{AC^0}$ reductions.

\emph{(Graph automorphism)} A lot of information on the complexity of graph automorphism can be found in the book by K{\"o}bler, Sch{\"o}ning, and Tor{\'a}n~\cite{MR1232421}. Concerning hardness of the automorphism problem, improving previous results of Tor{\'{a}}n~\cite{DBLP:journals/siamcomp/Toran04}, Wagner shows hardness results for graphs of bounded maximum degree~\cite{DBLP:conf/sofsem/Wagner08,WagnerThesis}. Agrawal and Arvind show truth table equivalence of several problems related to graph automorphism~\cite{DBLP:journals/iandc/AgrawalA96} and Arvind, Beigel, and Lozano study modular versions of graph automorphism~\cite{DBLP:journals/siamcomp/ArvindBL00} which for~$k\in\mathbb{N}$ ask whether the number of automorphisms of a given graph is divisible by~$k$.

The graph automorphism problem is of interest in quantum computing since it can be encoded as a hidden shift problem, as opposed to the graph isomorphism problem that is only known to be encodable as a hidden subgroup problem~\cite{DBLP:journals/qic/ChildsW07,DBLP:journals/siamcomp/HallgrenRT03}.

Recently, Allender, Grochow, and Moore~\cite{DBLP:journals/corr/AllenderGM15} developed a zero-error randomized reduction from graph automorphism to~$\LanguageTYPESET{MKTP}$, the problem of minimizing time-bounded Kolmogorov complexity, a variant of the minimum circuit size problem.
In that paper they also extend this to a bounded-error randomized reduction from graph isomorphism to~$\LanguageTYPESET{MKTP}$.

\emph{(Tournament isomorphism)} Concerning the tournament isomorphism problem, the currently fastest algorithm~\cite{DBLP:conf/stoc/BabaiL83} has a running time of~$n^{O(\log n)}$.
With respect to hardness, Wagner's results for tournament automorphism also apply to tournament isomorphism~\cite{DBLP:conf/mfcs/Wagner07}.

Ponomarenko showed that isomorphism of cyclic tournaments can be decided in polynomial time~\cite{Ponomarenko1994}, where a cyclic tournament is a tournament that has an automorphism that is a permutation with a single cycle spanning all vertices. Furthermore he showed that isomorphism of Schurian tournaments can be decided in polynomial time~\cite{Ponomarenko2013}.

\section{Sampling characteristic subsets.}\label{sec:sampling:subsets}

Let~$M$ be a finite set. We define a \emph{sampler~$\str$} over~$M$ to be a probability measure~$\Pr_{\str}\colon M\rightarrow [0,1]$ on the elements of~$M$. We think of a sampler as an oracle that we can invoke in order to obtain an element of~$M$. That is, given a sampler, we can sample a sequence of elements~$m_1,\ldots,m_t$ where each~$m_i$ is sampled independently from~$M$ according to~$\Pr_{\str}$.

We call a subset~$M'$ of~$M$ \emph{characteristic} with respect to~$\str$ if for all~$m,m'\in M$ it holds that~$m\in M'$ and~$\Pr_{\str}(m') = \Pr_{\str}(m)$ implies~$m'\in M'$. Another way of formulating this condition is that~$M'$ is invariant under all probability-preserving bijections~$\varphi\colon M \rightarrow M$, that is, those bijections that satisfy~$\Pr_{\str}(m) = \Pr_{\str}(\varphi(m))$ for all~$m\in M$. 

When considering sampling algorithms we will not assume that we know the size of the set~$M$. 
Our goal is to repeatedly invoke a sampler~$M$ so as to find a characteristic subset. The main difficulty in this is that we can never precisely determine the probability~$\Pr_{\str}(m)$ of an element~$m$. Indeed, the only thing we can hope for is to get a good estimate for such a probability. The following lemma indicates that this might be helpful since the set of probabilities cannot be arbitrarily dense.

\begin{lemma}\label{lem:some:empty:interval}
Let~$\Pr_{\str}$ be a discrete probability measure on the set~$M$. Let~$P = \{\Pr_{\str}(m)\mid m\in M\}$ be the set of probabilities that occur. For every positive integer~$i$ there is a~$j\in \{6i+1,\ldots,8i\}$ such that~$[ (j-1/4)/(8i^2),(j+1/4)/(8i^2)]\cap P = \emptyset$.
\end{lemma}

\begin{proof}
Suppose for all~$j\in \{6i+1,\ldots,8i\}$ there is some~$m_j$ with~$\Pr_{\str}(m_j) \in [ (j-1/4)/(8i^2),(j+1/4)/{8i^2}]$. Then~$\Pr_{\str}(m_j)\neq \Pr_{\str}(m_{j'})$ whenever~$j\neq j'$, implying in particular~$m_j \neq m_{j'}$. This yields~$2i$ distinct elements~$m_j$. Furthermore~$\Pr_{\str}(m_j)> 3/(4i)$  for all~$j\in \{6i+1,\ldots,8i\}$.
 Thus~$\Pr_{\str}(\{m_j\mid j\in\{6i+1,\ldots,8i\}\}) > 2i \cdot  3/(4i) >1$ yielding a contradiction.
\end{proof}

Using the lemma we can design an algorithm that, with high probability, succeeds at determining a characteristic set.

\begin{theorem}\label{thm:invariant:sampling}
There is a deterministic 
algorithm that, given~$\varepsilon>0$ and given access to a sampler~$S$ over an unknown set~$M$ of unknown size, runs in expected time polynomial in~$1/(\max_{m\in M}{\Pr_S(m)}) \leq |M|$ and~$\ln{1/{\varepsilon}}$ and outputs a non-empty subset of~$M$ that is characteristic with probability~$1-\varepsilon$.
\end{theorem}

\begin{proof}
Let~$p = \max_{m\in M}{\Pr_S(m)}$ and 
let~$i = \lceil 1/p\rceil$. First note that~$|M|\geq 1/p$ and that~$i \leq 2/p \leq 2|M|$.  Let~$P  = \{\Pr_S(m) \mid m\in M\}$ be the set of values that occur as probabilities of elements in~$M$. 

The idea of the proof is to sample many times as to get good estimates for probabilities using Chernoff bounds and then to include in the output all elements with a probability above a certain threshold. 
The main difficulty of the Lemma arises from the fact that~$p$ is not known to the algorithm. We first describe an algorithm for the situation in which~$p$ is known which works in such a way so that it can be adapted in the end.

We start by sampling~$T= \max\{\left\lceil i^3 2^{17} (\ln{1/{\varepsilon'}})\right\rceil, \left\lceil i^3 2^{18} (\ln{1/{\varepsilon'}})\right\rceil^2\}$ 
elements~$m_1,\ldots,m_T$ from the sampler, where we set~$\varepsilon' = \min\{1/e,\varepsilon/8\}$. 
 We then compute for each appearing element~$m_k$ a probability estimator~$\#(m_k)$ for its probability by computing~$N(m_k)/T$ where~$N(m_k)$ is the number of times that element~$m_k$ has been sampled. Let~$Q= \{\#(m_k)\mid k\in \{1,\ldots, T\}\}$ be the set of probability estimators. Let~$\ell$ be the smallest number in~$\{6i+1,\ldots,8i\}$ such that~$[ (\ell-1/8)/(8i^2),(\ell+1/8)/{(8i^2)}]\cap Q=\emptyset$. If no such element exists, we declare the algorithm as failed. 
Otherwise, we output~$M' = \{m_k \mid \#(m_k)> \ell/(8i^2)\}$. We call~$\ell$ the \emph{cut-off}.

\bigskip

We analyze the probability that this algorithm succeeds in computing a characteristic subset. For this, let us define~$\#(x) = 0$ for~$x\in M$ whenever~$x$ does not appear among the sampled elements.

\begin{claim}\label{claim:prob}
For each element~$x\in M$, the probability that~$|\#(x)-\Pr_S(x)| \geq 1/(2^7 i^2)$ is at most~$2e^{-\frac{T}{2^{17} i^3}} \leq 2\varepsilon'$.
\end{claim}
\proof
 Consider an experiment where we sample~$T$ elements according to~$S$. We want to bound the probability that the observed~$T\cdot \#(x)$ deviates from its expected value~$\mu \coloneqq T\cdot  \Pr_{\str}(x)$ by at least~$T/(2^7i^2)$.
This deviation is at least~$\delta\mu $ if we set~$\delta \coloneqq \frac{1}{2^7 i^2  \Pr_{\str}(x)} >0 $. 
We can thus use the Chernoff bound (see~\cite[Corollary A.15, Page 515]{DBLP:books/daglib/0023084}) and conclude that the probability that
$|\#(x)-\Pr_S(x)| \geq 1/(2^7 i^2)$
is at most
 \[2e^{-\mu \min\{\delta^2/4,\delta/2\} } \leq 
 2e^{\left(-T\min\{\frac{1}{2^{16} i^4 \Pr_{\str}(x)},\frac {1}{2^8 i^2}\} \right)}
  \leq 2e^{\left(\frac{-T}{\max\{2^{17} i^3,2^8 i^2 \} }\right)} \leq  
  2e^{\left(\frac{-T}{ 2^{17} i^3}\right)}, \]
  where the second inequality uses the fact~$\Pr_{\str}(x) \leq p = 2/(2/p) \leq 2/\lceil 1/p\rceil =2/i$.
\uend

Define~$A_k$ as the event that for the~$k$-th sampled element~$m_k$ we have~$|\#(m_k) - \Pr(m_k)| \geq 1/(2^6 i^2)$. Thus the event~$A_k$ happens if~$\#(m_k)$ deviates excessively from its expected value.

\begin{claim}[resume]\label{claim:Ak}
The probability that there is a~$k\in\{1,\ldots,T\}$ such that event~$A_k$ occurs  is at most~$2\varepsilon'$.
\end{claim}

\proof
To bound the probability of event~$A_k$, we first consider~$\Pr(A_k \mid m_k = x)$, the probability of~$A_k$ under the condition that the~$k$-th sampled element~$m_k$ is equal to~$x$ for some fixed element~$x\in M$. Considering that we already know that~$m_k = x$ we need to consider an experiment where we sample~$T-1$ times independently from~$S$ and count the number of times we obtain element~$x$. This number is then~$N(m_k)-1$ since the item with number~$k$ itself adds one to the count of elements equal to~$x$.
If~\[\#(m_k)  = N(m_k)/T \notin  [ \textstyle{\Pr_S(m_k)}-1/(2^6 i^2),\textstyle{\Pr_S(m_k)}+1/(2^6 i^2)]\] then~\[(N(m_k)-1)/(T-1) 
 \notin [ \textstyle{\Pr_S(m_k)}-1/(2^7 i^2),\textstyle{\Pr_S(m_k)}+1/(2^7 i^2)],\] as shown by the simple fact that for positive integers~$2\leq a\leq b$ we have~$|a/b - (a-1)/(b-1)| \leq 1 / (b-1)$ and~$1/(2^7 i^2 ) \geq 1/(T-1)$. 

Thus in our experiment with~$T-1$ trials, the observed value~$N(m_k)-1$ must deviate from its expected value~$\mu \coloneqq (T-1) \Pr_{\str}(x)$ by at least~$ (T-1)/(2^7 i^2)$.

From the previous claim we obtain an upper bound of\[2e^{-\frac{(T-1)}{ 2^{17} i^3}} \leq 2e^{-\frac{T}{ 2^{18} i^3}},\]
where the inequality uses the fact that~$T\geq 2$.

Since this bound is independent of~$x\in M$ and since~$x$ was arbitrary, the bound is also an upper for~$\Pr(A_k)$.
   By the union bound and using~$T\geq \left\lceil i^3 (\ln{1/{\varepsilon'}}) 2^{18} \right\rceil^2$, we obtain that the probability that there is a~$k\in\{1,\ldots,T\}$ such that~$A_k$ happens is at most~\[T \cdot  2e^{-\frac{T}{2^{18} i^3}}\leq T2e^{-\sqrt{T}} {\varepsilon'} \leq 2{\varepsilon'},\]
  where the last inequality follows since~$t^2e^{-t} <1$ for~$t\geq 1$.
  \uend

\begin{claim}[resume]\label{claim:failed}
If the algorithm is declared as failed then~$A_k$ occurs for some~$k\in \{1,\ldots,T\}$.
\end{claim}

\proof
By  Lemma~\ref{lem:some:empty:interval} there is an integer~$j\in \{6i+1,\ldots,8i\}$ such that~$\Pr_{\str}(m)\notin  [(j-1/4)/(8i^2),(j+1/4)/(8i^2)]$ for all~$m\in M$. Define~$B_k$ as the event that for the~$k$-th sampled element~$m_k$ we have~$\#(m_k) \in [ (j-1/8)/(8i^2),(j+1/8)/(8i^2)]$. The algorithm can only be declared a failure if event~$B_k$ happens for some~$k\in \{1,\ldots,T\}$. However, the event~$B_k$ implies the event~$A_k$.
\uend

\begin{claim}[resume]\label{claim:empty}
Assuming the algorithm is not declared a failure, the probability that~$M'$ is empty is at most~$2\varepsilon'$.
\end{claim}
\proof
Since~$i\geq \lceil 1/p\rceil $ there is an element~$x\in M$ with~$\Pr_{\str} (x) \geq 1/i \geq \ell/{(8i^2)}$. Then the probability that~$\#(m_i) < (\ell-1/8)/{(8i^2)} \leq (8i-1/8)/(8i^2) = (1-1/(64i))/i\leq (1-1/(64i)) \Pr_{\str} (x)$ is at most~$2\varepsilon'$ by Claim~\ref{claim:prob}.
\uend

\begin{claim}[resume]\label{claim:canonical}
If~$M'$ is not characteristic then~$A_k$ occurs for some~$k\in \{1,\ldots,T\}$ with probability at least~$(1-\varepsilon')$.
\end{claim}
\proof
By~Claim~\ref{claim:failed} we can assume that the algorithm was not declared a failure. Note that, if~$M'$ is not characteristic then one of the following three things happens: there is an element~$m_k$ with~$\#(m_k)> j$ but~$Pr_S(m_k)\leq \ell$ or there is an element~$m_k$ with~$\#(m_k)\leq \ell$ but~$Pr_S(m_k)> \ell$, or ~$\#(x) = 0$ for an element with~$Pr_S(x)> \ell$.  However, by the choice of~$\ell$, we know that~$\#(m_k)\notin [ (\ell-1/8)/(8i^2),(\ell+1/8)/{(8i^2)}]$. Thus in the first two cases we conclude that event~$A_k$ occurs.
The third option is that~$\#(x) = 0$ for an element with~$Pr_S(x)> \ell$. There are at most~$1/\ell<8i^2/(6i)= 4i/3$ elements~$x$ with such a probability and for each the probability for~$\#(x) = 0$ is at most~$(1-\ell)^T\leq (1-3/(4i))^T\leq \varepsilon' 3/(4i)$. So by the union bound we obtain a total probability of at most~$\varepsilon'$.

\uend

Combining the claims we obtain that the algorithm fails with probability at most~$2\varepsilon'+ 2\varepsilon' +\varepsilon \leq 5\varepsilon' \leq \varepsilon$.

Until this point we have assumed that the value of~$p$ is known to the algorithm. To remedy this we repeatedly run the algorithm with a simple doubling technique. In each iteration we run the described algorithm assuming that~$1/p\in [i,2i]$. Here we sample~$T$ elements of~$M$. In the next iterations we replace~$i$ by~$2i$ and repeat. We also replace~${\varepsilon'}$ by~${\varepsilon'} /2$. Since~$p\geq 1/|M|$, The number of iterations is logarithmic in~$1/p$.
The total number of sampled items is at most twice the number of items sampled in the last round. Thus, overall we obtain an algorithm with expected polynomial time.
To ensure that we obtain a suitable error bound it suffices to note that the probabilities of Claims~\ref{claim:prob},~\ref{claim:Ak} and~\ref{claim:canonical} actually decrease when~$i$ is replaced by an arbitrary smaller number. Skipping the first round, we obtain an error of at most~$5\varepsilon'/2 + 5\varepsilon'/4 +5\varepsilon'/8 +\ldots \leq \varepsilon$.
(Note that this argument in particular comprises the fact that if in an iteration a set is being output by the algorithm it is still characteristic with sufficiently high probability.)
\end{proof}

We note several crucial observations about any algorithm solving the problem just described.
There is no algorithm that for every set~$M$ and sampler~$S$ always outputs the same set~$M'$ with high probability.

Indeed, consider the set~$M = \{a,b\}$. Choosing~$\Pr_S (a) = \Pr_S(b) = 1/2$ means that~$M'$ must be~$\{a,b\}$.
Choosing~$\Pr_S(a) = 1$ and~$\Pr_{\str}(b) = 0$ implies that~$M'$ must be~$\{a\}$. However, there is a continuous deformation between these two samplers, while possibilities for the set~$M'$ are discrete. It is not difficult to see that the probability distribution of the output set~$M'$ must be continuous in the space of samplers, and thus, whatever the algorithm may be, there must be samplers for which the algorithm sometimes outputs~$\{a\}$ and sometimes outputs~$\{a,b\}$.

Let us also remark that the analysis of the running time of the algorithm is certainly far from optimal. In particular a large constant of~${(2^{18})}^2$ arises only from the goal to keep the computations simple and the desire to have a bound that also holds for small values of~$|M|$. 

Once one is interested in small running times, one might even ask whether it is possible to devise an algorithm running in time sublinear in~$|M|$. However, recalling the coupon collector theorem and considering uniform samplers one realizes that one cannot expect to make do with~$o(|M|\log |M|)$ samplings. However, if the set~$M$ is of algebraic nature, for example forms a group, then there might be meaningful ways to sample characteristic substructures (see Section~\ref{sec:open:prob}).

\section{Gadget constructions for asymmetric tournaments}\label{sec:gadget:constructs}

There are several computational problems fundamentally related to the graph isomorphism problem. This relation manifests formally as polynomial-time Turing (or even many-one) reductions between the computational tasks. Such reductions are typically based on gadget constructions which we revisit in this section.

While the \emph{graph isomorphism problem}~$\GIlang$ asks whether two given graphs are isomorphic, in the search version of this decision problem an explicit isomorphism is to be found, whenever one exits.
The \emph{graph automorphism problem}~$\GAlang$  asks whether a given graph has a non-trivial automorphism (i.e., an automorphism different from the identity). In other words the task is to decide whether the given graph is asymmetric. Two other related problems are the task~$\AUTlang$ to determine generators for the automorphism group~$\Aut(G)$ and the task to determine the size of the automorphism group~$|\Aut(G)|$.
 
For all named problems there is a colored variant, where the given graphs are vertex colored and isomorphisms are restricted to be color preserving. We denote the respective problems by~$\cGIlang$,~$\cGAlang$ and~$\cAUTlang$.

It is well known that between all these computational problems -- except~$\GAlang$ -- there are polynomial-time Turing reductions (we refer for example to \cite{boothcolbourn},~\cite{MR1232421}, \cite{DBLP:journals/ipl/Mathon79}). Concerning the special case of~$\GAlang$, while there is a reduction from~$\GAlang$ to the other problems, a reverse reduction is not known.

The reductions are typically stated for general graphs, but many of the techniques are readily applicable to restricted graph classes. By a \emph{graph class} we always mean a collection of possibly directed graphs closed under isomorphism. The \emph{isomorphism problem for graphs in~$\mathcal{C}$}, denoted~$\GIlang_\mathcal{C}$, is the computational task to decide whether two given input graphs from~$\mathcal{C}$ are isomorphic. If one of the input graphs is not in~$\mathcal{C}$ the answer of an algorithm may be arbitrary, in fact the algorithm may even run forever. 
Analogously, for each of the other computational problems that we just mentioned, we can define a problem restricted to~$\mathcal{C}$ giving us for example~$\GAlang_\mathcal{C}$, and $\AUTlang_\mathcal{C}$  and the colored versions~$\cGIlang_\mathcal{C}$,~$\cGAlang_\mathcal{C}$, and~$\cAUTlang_\mathcal{C}$.

As remarked in~\cite{DBLP:journals/jcss/ArvindDM10}, most of the reduction results for general graphs transfer to the problems for a graph class~$\mathcal{C}$ if one has, as essential tool, a reduction from~$\cGIlang_\mathcal{C}$ to~$\GIlang_\mathcal{C}$.

\begin{theorem}[Arvind, Das, Mukhopadhyay~\cite{DBLP:journals/jcss/ArvindDM10}]\label{thm:reductios:relative:to:class} Suppose that for a graph class~$\mathcal{C}$ there is a polynomial-time many-one reduction from $\cGIlang_\mathcal{C}$ to~$\GIlang_\mathcal{C}$ (i.e., $\cGIlang_\mathcal{C} \leq_m^p \GIlang_\mathcal{C}$)\footnote{Let us remark for completeness that a Turing reduction assumption~$\cGIlang_\mathcal{C} \leq_T^p \GIlang_\mathcal{C}$ actually suffices for the theorem.}. Then  
\begin{enumerate}
\item $\GAlang_\mathcal{C}$ polynomial-time Turing-reduces to~$\GIlang_\mathcal{C}$ (i.e.,~$\GAlang_\mathcal{C} \leq_T^p \GIlang_\mathcal{C}$),
\item The search version of~$\GIlang_\mathcal{C}$ polynomial-time Turing-reduces to the decision version of~$\GIlang_\mathcal{C}$, and
\item $\AUTlang_\mathcal{C}$  polynomial-time Turing-reduces to~$\GIlang_\mathcal{C}$ (i.e.,~$\AUTlang_\mathcal{C} \leq_T^p \GIlang_\mathcal{C}$).\label{item:aut:to:gi}
\end{enumerate}
\end{theorem}

In this paper we are mainly interested in two classes of directed graphs, namely the class of tournaments~$\Tourlang$ and the class of asymmetric tournaments~$\AsymTourlang$. For the former graph class, a reduction from the colored isomorphism problem to the uncolored isomorphism problem is given in~\cite{DBLP:journals/jcss/ArvindDM10}.

\begin{theorem}[Arvind, Das, Mukhopadhyay~\cite{DBLP:journals/jcss/ArvindDM10}]\label{thm:removing:colors:for:tour:iso}
The colored tournament isomorphism problem is polynomial-time many-one reducible to the (uncolored) tournament isomorphism problem (i.e.,~$\cGIlang_\mathcal{\Tourlang} \leq_m^p \GIlang_\mathcal{\Tourlang}$).
\end{theorem}

However, for our purposes we also need the equivalent statement for asymmetric tournaments. 
Taking a closer look at the reduction described in \cite{DBLP:journals/jcss/ArvindDM10} yields the desired result. In fact it also shows that the colored asymmetry problem reduces to the uncolored asymmetry problem. Denoting for a graph class~$\mathcal{C}$ by~$\Asymlang\mathcal{C}$ the class of those graphs in~$\mathcal{C}$ that are asymmetric (i.e., have a trivial automorphism group), we obtain the following.

\begin{lemma}
\label{lem:col:tour:iso:to:tour:asym}
\begin{enumerate}
\item The isomorphism problem for colored asymmetric tournaments is polynomial-time many-one reducible to the isomorphism problem for (uncolored) asymmetric tournaments (i.e.,~$\cGIlang_\mathcal{\AsymTourlang} \leq_m^p \GIlang_\mathcal{\AsymTourlang}$).

\item The colored tournament asymmetry problem is polynomial-time many-one reducible to the (uncolored) tournament asymmetry problem (i.e.,~$\cGAlang_\mathcal{\Tourlang} \leq_m^p \GAlang_\mathcal{\Tourlang}$).
\end{enumerate}

\end{lemma}

\begin{proof}[Proof sketch]
In~\cite{DBLP:journals/jcss/ArvindDM10} given two colored tournaments~$T_1$ and~$T_2$, a gadget construction is described that adds new vertices to each tournament yielding~$T_1'$ and~$T_2'$ so that~$T_1\cong T_2 \Leftrightarrow  T'_1\cong T'_2$. 
The authors show that every automorphism of~$T'_i$ fixes the newly added vertices. However, from the construction it is clear that~$T_1$ is asymmetric if and only if~$T'_i$ is asymmetric, since all vertices that are added must be fixed by every automorphism. This demonstrates both parts of the lemma. 
We sketch a gadget construction that achieves these properties and leave the rest to the reader. For each~$i\in \{1,2\}$ the construction is as follows. Suppose without loss of generality that the colors of~$T_i$ are~$\{1,\ldots,\ell\}$ with~$\ell\geq 2$.
We add a directed path~$u_1\rightarrow \ldots\rightarrow u_{\ell}$ to the graph. A vertex~$v\in V(T_i)$ has~$u_j$ as in-neighbor if~$j$ is the color of~$v$. Otherwise~$u_j$ is an out-neighbor of~$v$. We add two more vertices~$a$ and~$,b$ to the graph. The only out-neighbor of vertex~$a$ is~$b$. The in-neighbors of~$b$ are the vertices in~$\{a,u_1,\ldots,u_{\ell}\}$. It can be shown that~$a$ is the unique vertex with maximum in-degree. This implies that~$b$ and thus all~$u_j$ are fixed by all automorphisms.
\end{proof}

As mentioned above, reductions for computational problems on general graphs can often be transferred to the equivalent problems restricted to a graph class~$\mathcal{C}$. However, let us highlight a particular reduction where this is not the case. Indeed, it is not clear how to transfer the reduction from~$\GIlang$ to~$\AUTlang$ (which involves taking unions of graphs) to a reduction from~$\GIlang_\mathcal{C}$ to~$\AUTlang_\mathcal{C}$, even when provided a reduction of~$\cGAlang_\mathcal{C}$ to~$\GIlang_\mathcal{C}$. 
For the class of tournaments however, we can find such a reduction, of which we can make further use. 

\begin{lemma}\label{lem:asym:iso:red:to:asym} 
\begin{enumerate}
\item  The isomorphism problem for tournaments polynomial-time  Turing-reduces to
the task to compute a generating set for the automorphism group of a tournament (i.e., $\cGIlang_\mathcal{\Tourlang} \leq_T^p \AUTlang_\mathcal{\Tourlang}$). \label{item:gi:to:aut}

\item The isomorphism problem for colored asymmetric tournaments is polynomial-time many-one  reducible to tournament asymmetry (i.e., $\cGIlang_\mathcal{\AsymTourlang} \leq_m^p \GAlang_\mathcal{\Tourlang}$). \label{item:asym:iso:to:asym}
\item The search version of the isomorphism problem for colored asymmetric tournaments Turing-reduces to tournament asymmetry.\label{item:asym:iso:to:asym:search}
\end{enumerate}
\end{lemma}

\begin{figure}
\centering
\scalebox{0.9}{
\begin{tikzpicture}
\node[circle, draw, inner sep=10, outer sep =3] (T1) at (-0,2) {$T_1$};
\node[circle, draw, inner sep=10, outer sep =3] (T2) at (-4,0) {$T_2$};
\node[circle, draw, inner sep=10, outer sep =3] (T1p) at (-0,-2) {$T'_1$};

\draw[transform canvas={yshift=0.3ex},->, ultra thick] (T1.290) -- (T1p.70);
\draw[transform canvas={yshift=0.3ex},->, ultra thick] (T1.270) -- (T1p.90);
\draw[transform canvas={yshift=0.3ex},->, ultra thick] (T1.250) -- (T1p.110);

\draw[transform canvas={yshift=0.3ex},->, ultra thick] (T1p.160) -- (T2.-30);
\draw[transform canvas={yshift=0.3ex},->, ultra thick] (T1p.180) -- (T2.-50);
\draw[transform canvas={yshift=0.3ex},->, ultra thick] (T1p.140) -- (T2.-10);

\draw[transform canvas={yshift=0.3ex},->, ultra thick] (T2.30) -- (T1.-160);
\draw[transform canvas={yshift=0.3ex},->, ultra thick] (T2.50) -- (T1.-180);
\draw[transform canvas={yshift=0.3ex},->, ultra thick] (T2.10) -- (T1.-140);

\node at (2.5,-2) {$T_1\cong T_2$};
\end{tikzpicture}}
\caption{A visualization of the triangle tournament~$\mathrm{Tri}(T_1,T_2)$.}
\label{fig:tri}
\end{figure}
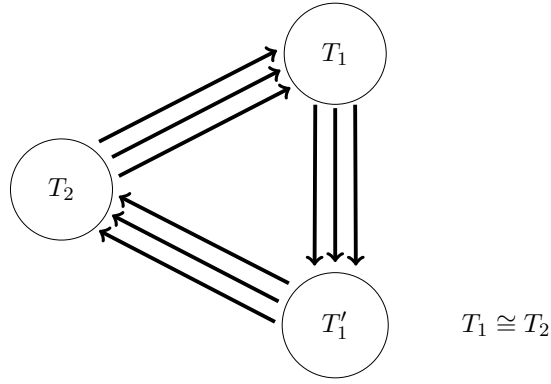
\begin{proof}
Suppose we are given two tournaments~$T_1$ and~$T_2$ on the same number of vertices~$n$ for which isomorphism is to be decided. By Theorem~\ref{thm:removing:colors:for:tour:iso} we can assume that the tournaments are uncolored. 
Let~$\mathrm{Tri}(T_1,T_2)$ be the tournament obtained by forming the disjoint union of the three tournaments~$T_1$,~$T_1'$ and~$T_2$ where~$T_1\cong T_1'$ .
We add edges from all vertices of~$T_1$ to all vertices of~$T_1'$, from all vertices of~$T_1'$ to all vertices of~$T_2$ and from all vertices of~$T_2$ to all vertices of~$T_1$ (see Figure~\ref{fig:tri}).
We observe that two vertices that are contained in the same of the three sets~$V(T_1)$,~$V(T_2)$,~$V(T_1')$ have~$n$ common out-neighbors. However, two vertices that are not contained in the same of these three sets have at most~$n-1$ common out-neighbors.
We conclude that an automorphism of~$\mathrm{Tri}(T_1,T_2)$ preserves the partition of~$V(\mathrm{Tri}(T_1,T_2))$ into the three sets~$V(T_1)$,~$V(T_1')$ and~$V(T_2)$. 

Given a generating set for~$\Aut(\mathrm{Tri}(T_1,T_2))$ it holds that there is some generator that maps a vertex from~$V(T_1)$ to a vertex from~$V(T_2)$ if and only if~$T_1$ and~$T_2$ are isomorphic. This proves the first part of the lemma.

Suppose additionally that~$T_1$ and~$T_2$ are asymmetric.
We then further conclude that the tournament~$\mathrm{Tri}(T_1,T_2)$ has a non-trivial automorphism if and only if~$T_1$ and~$T_2$ are isomorphic. This shows that the decision version of asymmetric tournament isomorphism reduces to tournament asymmetry.
Since the search version is Turing-reducible to the decision version of isomorphism (Theorem~\ref{thm:reductios:relative:to:class}) this finishes the proof. 
\end{proof}

For Turing reductions, the converse of the previous lemma also holds.
 In fact the converse holds for arbitrary graph classes. 

\begin{lemma} \label{lem:asym:red:to:asym:iso} 
Let~$\mathcal{C}$ be a graph class.
\begin{enumerate}
\item  The task to compute a generating set for the automorphism group of graphs in~$\mathcal{C}$ Turing-reduces to the isomorphism problem for colored graphs in~$\mathcal{C}$ (i.e.,~$\AUTlang_\mathcal{\mathcal{C}} \leq_T^p \cGIlang_\mathcal{\mathcal{C}}$).

\item Asymmetry checking for graphs in~$\mathcal{C}$ polynomial-time Turing-reduces to isomorphism checking of asymmetric colored graphs in~$\mathcal{C}$ (i.e., 
 $\GAlang_\mathcal{\mathcal{C}} \leq_T^p \cGIlang_\mathcal{\Asymlang\mathcal{C}}$).
\end{enumerate}
\end{lemma}

\begin{proof}

The proof of the first part is a well known reduction that already appears in~\cite{DBLP:journals/ipl/Mathon79}. We can also see it by applying Part~\ref{item:aut:to:gi} of Theorem~\ref{thm:reductios:relative:to:class} to the class of colored graphs in~$\mathcal{C}$. 

For the second part, assume we have an oracle~$O_1$ for isomorphism checking of colored asymmetric graphs in~$\mathcal{C}$. Then we also have an oracle~$O_2$ for the search-version of isomorphism checking of colored asymmetric graphs in~$\mathcal{C}$. Indeed, we can find an isomorphism by individualizing more and more vertices in both graphs while keeping the graphs isomorphic. When all vertices are singletons, there is only one option for the isomorphism.

Now let~$G$ be a graph in~$\mathcal{C}$. Without loss of generality assume that~$V(G)  = \{v_1,\ldots,v_n\}$.
For every~$t,t'\in \{1,\ldots,n\}$ with~$t>t'$ we call~$O_2(G_{(v_1,\ldots,v_{t-1},v_t)},G_{(v_1,\ldots,v_{t-1},v_{t'})})$. 
Here the notation~$G_{(u_1,\ldots,u_{\ell})}$ denotes the graph~$G$ colored such that the color of~$u_i$ is~$i$ and vertices not in~$\{u_1,\ldots,u_{\ell}\}$ have color 0. (With respect to the partition of the vertices into color classes this is the same as constructing the graph obtained from~$G$ by individualizing~$u_1,\ldots,u_{\ell}$ one after the other.)
If we find an isomorphism among the calls then this isomorphism is non-trivial since it maps~$v_t$ to~$v_{t'}$ and thus~$G$ is not asymmetric. Conversely if~$G$ is not asymmetric, then let~$j$ be the least integers for which~$G_{(v_1,\ldots,v_{j})}$ is asymmetric. Then~$j<n$, (since~$G_{(v_1,\ldots,v_{n-1})}$ is always asymmetric) and there is a~$t'>j$ such that~$G_{(v_1,\ldots,v_{j-1},v_{j})}$ and~$G_{(v_1,\ldots,v_{j-1},v_{t'})}$ are isomorphic. This isomorphism will be found by the oracle.

While the oracle~$O_2$ can sometimes output incorrect answers, namely when one of the inputs is not asymmetric, $O_2$ is certifying in the sense that we can check whether a given answer is really an isomorphism. Thus, we avoid making any errors whatsoever.
\end{proof}

\section{Invariant automorphism samplers from asymmetry tests}

As discussed before, the asymmetry problem of a class of graphs reduces to the isomorphism problem of graphs in this class. However, whether there is a reduction in the reverse, or whether the asymmetry problem may actually be computationally easier than the isomorphism problem is not known. To approach this question, we now explore what computational power we could get from having available an oracle for the asymmetry problem.

An \emph{invariant automorphism sampler for a graph~$G$} is a sampler over~$\Aut(G)\setminus\{\id\}$ which satisfies the property that if~$\Pr_{\str}(\varphi) = p$ then~$\Pr_{\str}(\psi^{-1}\circ \varphi \circ \psi) = p$ for all~$\psi \in \Aut(G)$.
We first show how to use an oracle for asymmetry to design an invariant automorphism sampler for a tournament~$T$. 

\begin{lemma}\label{lem:from:tour:asym:to:aut:sampler}
Given an oracle for asymmetry of tournaments ($\GAlang_\mathcal{\Tourlang}$)
we can construct for every given colored (or uncolored) tournament~$T$ that is not asymmetric an invariant automorphism sampler.
The computation time (and thus the number of oracle calls) required to sample once from~$\str$  is polynomial in~$|V(T)|$.
\end{lemma}

\begin{algorithm}
	\caption{An invariant automorphism sampler for tournaments using an asymmetry oracle}\label{alg:iso:sampler:from:asym:oracle}
	
    \label{alg:algorithm-label}
    \begin{algorithmic}[1]
    \REQUIRE A tournament~$T$ that is not asymmetric and an oracle~$O$ for tournament asymmetry.
    	\ENSURE An automorphism~$\varphi\in \Aut(T)\setminus\{\id\}$. As a random variable, the outputs of the algorithm form an invariant automorphism sampler for~$T$.
    	\ENSUREGAP
         \STATE $T_{\mathrm{next}} \leftarrow T$
         \WHILE{$\Aut(T_{\mathrm{next}})\neq \{\id\}$} 
         \STATE{Pick a vertex~$v$ independently, uniformly at random among all non-singleton color classes in~$T_{\mathrm{next}}$.}
         \STATE $T\leftarrow T_{\mathrm{next}}$  
         \STATE{$T_{\mathrm{next}}\leftarrow T_{(v)}$}\COMMENT{individualize~$v$}
         \ENDWHILE
         \COMMENT{at this point~$T_{\mathrm{next}}$ is asymmetric}
         \STATE Let~$V'$ be the set of those vertices that have the same color in~$T$ as~$v$.
         \STATE Let~$V''$ be the set of those vertices~$v''$ in~$V'\setminus \{v\}$ for which~$\Aut(T_{(v'')})= \{\id\}$.
         \STATE Let~$V'''$ be the set of those vertices~$v'''$ in~$V''$ for which
          $T_{\mathrm{next}}\cong T_{(v''')}$. \\\COMMENT{use Part~\ref{item:asym:iso:to:asym} of Lemma~\ref{lem:asym:iso:red:to:asym} }
         \STATE Pick a vertex~$u\in V'''$ uniformly at random.
         \STATE Compute an isomorphism~$\varphi$ from~$T_{\mathrm{next}}$ to~$T_{(u)}$. \COMMENT {there is only one such isomorphism}
         \RETURN $\varphi$
    \end{algorithmic}
\end{algorithm}
\begin{proof}
Let~$O_1$ be an oracle for uncolored tournament asymmetry.
By Lemma~\ref{lem:col:tour:iso:to:tour:asym}, we can transform the oracle~$O_1$ for the asymmetry of uncolored  tournaments into an oracle~$O_2$ for asymmetry of colored tournaments. By Lemma~\ref{lem:asym:iso:red:to:asym} Part~\ref{item:asym:iso:to:asym}, we can also assume that we have an oracle~$O_3$ that decides the isomorphism problem of colored asymmetric tournaments.
More strongly, Lemma~\ref{item:asym:iso:to:asym:search} Part~\ref{item:asym:iso:to:asym} makes a remark on the search version, thus we can assume that~$O_3$ also solves the isomorphism search problem for asymmetric tournaments.

To obtain the desired sampler~$\str$ we proceed as follows. In the given tournament~$T$ we repeatedly fix (by individualization, i.e., giving it a special color) uniformly, independently at random more and more vertices until the resulting tournament is asymmetric. This gives us a sequence of colored tournaments~$T = T_0, T_1, \ldots, T_t$ such that~$\Aut(T_t)  = \{\id \}$,~$\Aut(T_{t-1})  \neq \{\id \}$  and such that~$T_t = ({T_{t-1}})_{(v)}$ for some vertex~$v$. In other words,~$T_t$ is obtained from~$T_{t-1}$ by individualizing~$v$ which makes the graph asymmetric.
Using the available oracle~$O_2$, we can compute the set~$V''$ of those vertices~$v''$ in~$V(T)\setminus\{v\}$ that have the same color as~$v$ such that~$\Aut((T_{t-1})_{(v'')})= \{\id\}$. There must be at least one vertex in~$V''$ since~$T_{t-1}$ is not asymmetric.
Using the oracle~$O_3$, we can then compute the subset~$V'''\subseteq V''$ of those vertices~$v'''$ for which~$(T_{t-1})_{(v''')}$ and~$T_t$ are isomorphic. 
Next, we pick a vertex~$u\in V'''$ uniformly at random. Since both~$(T_{t-1})_{(u)}$  and~$T_{t}$ are asymmetric, using the oracle~$O_3$ for the isomorphism search problem we can compute an isomorphism~$\varphi$ from~$(T_{t-1})_{(u)}$ to~$T_{t}$. This isomorphism~$\varphi$ is unique and it is a non-trivial automorphism of~$\Aut(T)$. Algorithm~\ref{alg:iso:sampler:from:asym:oracle} gives further details.

\emph{(Invariance)} The invariance follows directly from the fact that all steps of the algorithm either consist of choosing a vertex uniformly at random or computing an object that is invariant with respect to all automorphisms fixing all vertices that have been randomly chosen up to this point.

\emph{(Running time)} Concerning the running time, one call of Algorithm~\ref{alg:iso:sampler:from:asym:oracle} uses less than~$2n$ calls to oracle~$O_2$ and at most~$n$ calls to oracle~$O_3$. The overall running time is thus polynomial.
\end{proof}

Let us comment on whether the technique of the lemma can be applied to graph classes other than tournaments. For the technique to apply to a graph class~$\mathcal{C}$, we require the oracle~$O_2$, which solves colored asymmetry~$\mathcal{C}$, and the oracle~$O_3$ which solves the isomorphism search problem for asymmetric colored objects in~$\mathcal{C}$. (The oracle~$O_1$ is a special case of~$O_2$.)
In the case of tournaments, having an oracle~$O_1$ (i.e., an oracle for uncolored asymmetry) is sufficient to simulate the oracles~$O_2$ and~$O_3$, but this is not necessarily possible for all graph classes~$\mathcal{C}$. It is however possible to simulate such oracles for every graph class that satisfy some suitable (mild) assumptions, as can be seen from the discussion in Section~\ref{sec:gadget:constructs}. In particular, given an oracle for asymmetry of all graphs we can construct an invariant automorphism sampler for all graphs that are not asymmetric.

\section{Invariant suborbits from invariant automorphism samplers}\label{sec:sampling:minimimal:orbits}

Let~$G$ be a directed graph. 
Let~$\str$ be an invariant automorphism sampler for~$G$.
We now describe an algorithm that, given access to an asymmetry oracle, constructs a non-discrete partition of~$V(G)$ which is finer than or at least as fine as the orbit partition of~$G$ under~$\Aut(G)$ and invariant under~$\Aut(G)$. Here, a partition~$\pi$ is invariant under~$\Aut(G)$ if~$\pi = \psi(\pi)$ for all~$\psi \in \Aut(G)$. (A partition is discrete if it consists only of singletons.)

\begin{theorem}\label{thm:from:aut:sampler:to:suborbits}
For every~$c\in \mathbb{N}$, there is a randomized polynomial-time algorithm that, given a graph~$G$ and an invariant automorphism sampler~$\str$ for~$G$ constructs with error probability at most~$\frac{1}{|G|^c}$ a non-discrete partition~$\pi$ of~$V(G)$ such that
\begin{enumerate}
\item $\pi$ is finer than or at least as fine as the orbit partition of~$V(G)$ under~$\Aut(G)$ and
\item $\pi$ is invariant under~$\Aut(G)$.
\end{enumerate}
The algorithm also provides a set of certificates~$\Phi = \{\varphi_1,\ldots, \varphi_m\} \subseteq \Aut(G)$
such that for every pair of vertices~$v,v'\in V(G)$ that lie in the same class of~$\pi$ there is some~$\varphi_i$ with~$\varphi_i(v) = v'$. 
\end{theorem}

\begin{proof}
Let~$M = \{(v,w)\mid v,w\in V(G), v\neq w, \exists \varphi\in \Aut(G)\colon  \varphi(v) = w \}$ be the set of pairs of two distinct vertices lying in the same orbit.
With the sampler~$\str$ we can simulate a sampler~$\str'$ over~$M$ invariant under~$\Aut(G)$ as follows. To create an element for~$\str'$ we sample an element~$\varphi$ from~$\str$ and uniformly at random choose an element~$v$ from the support~$\supp(\varphi) =\{x\in V(G)\mid \varphi(x)\neq x\}$ of~$\varphi$. Then the element for~$\str'$ is~$(v,\varphi(v))$.  It follows form the construction that~$\str'$ is a sampler for~$M$. Moreover, since all random choices are independent and uniform,~$\str'$ is invariant under automorphisms.

Using the algorithm from Theorem~\ref{thm:invariant:sampling} we can thus compute a characteristic subset~$M'$ of~$M$. Since~$\str'$ is $\Aut(G)$-invariant, the fact that~$M'$ is characteristic implies that it is also $\Aut(G)$-invariant. 
For the given~$c\in \mathbb{N}$, to obtain the right error bound, we choose~$\varepsilon$ to be~$\frac{1}{|G|^c}$ for the algorithm from Theorem~\ref{thm:invariant:sampling}. Then the error probability is at most~$\varepsilon= \frac{1}{|G|^c}$ and the running time is polynomial in~$|M| = O(|G|^2)$ and~$\ln |G|^c = O(|G|)$ and thus polynomial in the size of the graph.

Regarding~$M'$ as a binary relation on~$V(G)$ we compute the transitive closure and 
let~$\pi$ be the partition of~$V(G)$ into equivalence classes of said closure, where vertices that do not appear at all as entries in~$M'$ form their own class. By construction, elements that are in the same class of~$\pi$ are in the same orbit under~$\Aut(G)$. Moreover~$\pi$ is~$\Aut(G)$-invariant since~$M'$ is~$\Aut(G)$-invariant.

To provide certificates for the elements in~$M'$ we can store all elements given to us by~$S$. For each~$(v,w)\in M'$ we can thus compute an automorphism of~$\varphi_{v,w}\in \Aut(G)$ with~$\varphi_{v,w}(v) = w$. 
For pairs in the transitive closure of~$M'$ we then multiply suitable automorphisms.
\end{proof}

If a partition~$\pi$ satisfies the conclusion of the lemma, we call it an \emph{invariant collection of suborbits}. We call the elements of~$\Phi$ the \emph{certificates}. Let us caution the reader that the set~$\Phi$ returned by the algorithm is not necessarily characteristic. Moreover, the orbits of the elements in~$\Phi$ might not necessarily be contained within classes of~$\pi$. 
We call an algorithm an \emph{oracle for invariant suborbits} if, given a tournament~$T$, the algorithm returns a pair~$(\pi,\Phi)$ constituting invariant suborbits and certificates, in case~$T$ is not asymmetric, and returns the discrete partition~$\pi$ and~$\Phi =\{\id\}$ whenever~$T$ is asymmetric.

\section{Computing the automorphism group from invariant suborbits}\label{sec:auto:group:from:suborbits}

To exploit invariant suborbits we make use of the powerful group-theoretic technique to compute stabilizer subgroups.

\begin{theorem}[Luks~\cite{DBLP:journals/jcss/Luks82}]\label{thm:solvable:intersect}
There is an algorithm that, 
given a permutation group~$\Gamma$ on~$\{1,\ldots,n\}$ and subset~$B\subseteq   \{1,\ldots,n\}$, computes (generators for) the setwise stabilizer of~$B$. If~$\Gamma$ is solvable, then this algorithm runs in polynomial time.
\end{theorem}

We will apply the theorem in the following form: Let~$G$ be a graph and~$\Gamma$ a solvable permutation group on~$V(G)$. Then~$\Gamma \cap \Aut(G)$ can be computed in polynomial time. This follows directly from the theorem by considering the induced action of~$\Gamma$ on pairs of vertices from~$V(G)$ and noting that~$\Gamma \cap \Aut(G)$ consists of those elements that stabilize the edge set.

In our algorithm we will also use the concept of a quotient tournament (that can for example implicitly be found in~\cite{DBLP:journals/jcss/ArvindDM10}, see also~\cite{DBLP:conf/stacs/Schweitzer15}). Let~$T$ be a tournament and let~$\pi$ be a partition of~$V(T)$ in which all parts have odd size.
We define~$T/\pi$, the \emph{quotient of~$T$ modulo~$\pi$}, to be the tournament on~$\pi$ (i.e., the vertices of~$T/\pi$ are the parts of~$\pi$) where for distinct~$C,C'\in V(T/\pi) = \pi$ there is an edge from~$C$ to~$C'$ if and only if in~$T$ there are more edges going from~$C$ to~$C'$ than edges going from~$C'$ to~$C$. Note that since both~$|C|$ and~$|C'|$ are odd there are either more edges going from~$C$ to~$C'$ or more edges going from~$C'$ to~$C$. This implies that~$T/\pi$ is a tournament.
\begin{theorem}\label{thm:invariant:suborbits:give:tour:iso}
Suppose we are given as an oracle a randomized Las Vegas algorithm that computes invariant suborbits for tournaments in polynomial time. Then we can compute the automorphism group of tournaments in polynomial time.

\end{theorem}

\begin{algorithm}
	\caption{Computing the automorphism of a tournament using invariant suborbits}\label{alg:aut:group}
	
    \label{alg:auto:group}
    \begin{algorithmic}[1]
    \REQUIRE A (colored) tournament~$T$ and an oracle~$O$ for invariant suborbits with certificates.
    	
    \ENSURE A generating set for the automorphism group~$\Aut(T)$. 
    \ENSUREGAP
    \IF [Case 0]{$T$ is not monochromatic} \label{monocrome:case}
    \STATE Let $\mathrm{Col}$ be the set of vertex colors of~$T$.
    \FOR{$c\in COL$}
    \STATE Let $V^c$ be the set vertices in~$T$ of color~$c$.
    \STATE $\Psi^c\leftarrow  \Aut(T[V^c])$   \COMMENT{recursion}\label{item:case:0:recursion}
    \STATE Let $\widehat{\Psi^c}$ be the set of extensions of $\Psi^c$ to~$V(T)$ obtained by fixing vertices outside~$V^c$.
    \ENDFOR

    \STATE $\Psi = \bigcup_{c\in \mathrm{Col}} \widehat{\Psi^c}$
    \RETURN $\langle \Psi\rangle\cap \Aut(T)$ \COMMENT {solvable group stabilizer}\label{item:case:0:stab}
     \ENDIF   
    \STATE $(\pi,\Phi)\leftarrow O(T)$ \COMMENT {$\pi$ forms invariant suborbits of~$T$, $\Phi$ the set of certificates}

    \IF [$T$ is asymmetric]{$\pi$ is discrete}
    \RETURN $\{\id\}$
    \ELSIF [Case 1]{$\pi =  \{V(T)\}$}
    \STATE Choose $v\in V(T)$ arbitrarily.
    \STATE Let~$T'$ be obtained from~$T$ by coloring~$v$ with~1, all in-neighbors of~$v$ with 2 and other vertices with~$3$.
    \RETURN $\Phi \cup \Aut(V(T'))$  \COMMENT{recursion}\label{item:case:1:recursion}
    
    \ELSIF [Case 2]{$\exists C,C'\in \pi\colon |C|\neq |C'|$}
    \STATE Let~$T'$ be obtained from~$T$ by coloring each vertex~$v$ with color~$|[v]_{\pi}|$.
    \RETURN $\Aut(V(T'))$ \COMMENT{recursion}\label{item:case:2:recursion}
    
    \ELSE [Case 3]
        \STATE For~$C\in \pi$ we let~$T_C$ be the graph obtained from~$T[C]$ by picking an arbitrary vertex~$v\in C$ and coloring~$v$ with~1, all in-neighbors of~$v$ with 2 and other vertices with~$3$.\label{item:individ:tri}
        \FOR {$\{(C,C')\in \pi\mid C\neq C'\}$}
	        \STATE Compute~$\Aut (\mathrm{Tri}(T_C,T_C'))$ and extract an isomorphism $\varphi_{(C,C')}\colon T[C]\rightarrow T[C']$ whenever such an isomorphism exists.\COMMENT{recursion}\label{line:iso:call:via:aut}	   
        \ENDFOR
        \IF [Case 3a] {$\exists C,C'\in \pi\colon T[C] \ncong T[C']$}
	         \STATE Let~$T'$ be obtained from~$T$ by coloring~$V(T)$ so that~$v$ and~$v'$ have the same color if and only if~$T[([v])] \cong T[([v'])]$.
	         \RETURN $ \Aut(T')$ \COMMENT{recursion}\label{item:case:3a:recursion}
        \ELSE [Case 3b]
        \STATE $\Psi \leftarrow \Aut(T/\pi)$ \COMMENT{recursion on the quotient}\label{item:quotient}
        \STATE $\widehat{\Psi}\leftarrow \{\widehat{g} \mid g\in \Psi\}$, where~$\widehat{g}(v) =\varphi_{([v],g([v]))}(v)$.
        \FOR {$\{C \in \pi \}$}
        \STATE $\Upsilon_C \leftarrow \Aut(T[C])$ \COMMENT{recursion}\label{item:auts:of:parts}
        \STATE Compute~$\widehat{\Upsilon}_C$ the lifts of elements in~$\Upsilon_C$ by fixing vertices outside~$C$. 
        \ENDFOR
        \RETURN $\langle \widehat{\Psi}\cup  \bigcup_{C\in \pi} \widehat{\Upsilon}_C\rangle \cap \Aut(T)$   \COMMENT {solvable group stabilizer}\label{item:case:3b:stab}
       
        \ENDIF
        
    \ENDIF
    \end{algorithmic}
\end{algorithm}

\begin{proof}
We describe an algorithm that computes the automorphism group of a colored tournament given a randomized oracle that provides invariant suborbits.

\emph{(Description of the algorithm)}  Let~$T$ be a given colored tournament. 

(Case 0: $T$ is not monochromatic.) If~$T$ is not monochromatic then we proceed as follows:

Let~$\mathrm{Col}$ be the set of colors that appear in~$T$.
For~$c\in \mathrm{Col}$, let~$V^c$ be the set of vertices of color~$c$ and let~$T^c = T[V^c]$ be the subtournament induced by the vertices in~$V^c$.

We recursively compute~$\Aut(T^c)$ for all~$c\in \mathrm{Col}$. Let~$\Psi^c$ be the set of generators obtained as an answer. We lift every generator to a permutation of~$V(T)$ by fixing all vertices outside of~$V^c$. Let~$\widehat{\Psi^c}$ be the set of lifted generators of~$\Psi^c$ and let~$\Psi = \bigcup_{c\in \mathrm{Col}} \widehat{\Psi^c}$ be the set of all lifted generators.
Since~$\Aut(T^c) =\langle \Psi^c\rangle$ is solvable, we conclude that~$\langle \Psi\rangle$ is a direct product of solvable groups and thus solvable. We can thus compute~$\langle \Psi\rangle\cap \Aut(T)$ using Theorem~\ref{thm:solvable:intersect} and return the answer.

\medskip

This concludes Case 0. In every other case we first compute 
 a partition~$\pi$ into suborbits using the oracle and a corresponding set of certificates~$\Phi$.
For a partition~$\pi$ of some set~$V$ we denote for~$v\in V$ by~$[v]_{\pi}$ the element of~$\pi$ containing~$v$. We may drop the index when it is obvious from the context.
If~$|T|=1$ then we simply return the identity.

(Case 1: $\pi$ is trivial). In case~$\pi$ is trivial (i.e.,~$\pi=\{V(T)\}$), we know that~$T$ is transitive. We choose an arbitrary vertex~$v\in V(T)$. 
Let~$\lambda$ be the coloring of~$V(T)$ satisfying
\[\lambda(u) = \begin{cases}
1 & \text {if } u=v\\
2 & \text {if } (u,v)\in E(T)\\
3 & \text{otherwise}.
\end{cases}\]
We recursively compute a generating set~$\Psi$ for~$\Aut(T')$, where~$T'$ is~$T$ recolored with~$\lambda$. We then return~$\Psi \cup \Phi$. 

(Case 2: not all classes of $\pi$ have the same size.)

We color every vertex with the size of the class of~$\pi$ in which it is contained. Now~$T$ is not monochromatic anymore and we recursively compute~$\Aut(T)$ with~$T$ having said coloring. (In other words,
we proceed as in Case~$0$.)

(Case 3: all classes of $\pi$ have the same size but~$\pi$ is non-trivial.)

We compute for each pair of distinct equivalence classes~$C$ and~$C'$ of~$\pi$ an isomorphism~$\varphi_{(C,C')}$ from~$T[C]$ to~$T[C']$ or determine that no such isomorphism exists, as follows: We choose for each~$C$ an arbitrary vertex~$v\in C$.
We let~$T_C$ be the tournament obtained from~$T[C]$ by
coloring~$v$ with~1, all in-neighbors of~$v$ with 2 and other vertices with~$3$. We let~$T_{C,C'} = \mathrm{Tri}(T_C,T_{C'})$ be the triangle tournament of~$T_C$ and~$T_{C'}$ where~$(T_{C})'$ is an isomorphic copy of~$T_C$ (as defined in Section~\ref{sec:gadget:constructs} in the proof of Lemma~\ref{lem:asym:iso:red:to:asym}).

Using recursion we compute~$\Aut(T_{C,C'})$. From the result we can extract an isomorphism from~$T[C]$ to~$T[C']$ since~$V(T[C])$ and~$V(T[C'])$ are blocks of~$T_{C,C'}$.

(Case 3a:) If it is not the case that for every pair~$C,C'$ of color classes there is an isomorphism from~$T[C]$ to~$T[C']$ then we color the vertices of~$T$ so that~$v,v'$ have the same color if and only if there is an isomorphism from~$T[([v])]$ to~$T[([v'])]$, where as before for every vertex~$u$ we denote by~$[u]$ the class of~$\pi$ containing~$u$.
With this coloring,~$T$ is not monochromatic anymore and we recursively compute~$\Aut(T)$ with~$T$ having said coloring. (In other words,
we proceed as in Case~$0$.)

(Case 3b:) Otherwise, for every pair~$C,C'$ of color classes, there is an isomorphism from~$T[C]$ to~$T[C']$. Note that all color classes are of odd size since~$T[C]$ is transitive (as dictated by~$\pi$).
Thus, we can compute the quotient tournament~$T/\pi$. We recursively compute a generating set~$\Psi = \{ g_1,\ldots,g_t \}$ for the automorphism group of~$T/\pi$.

We lift each~$g_i$ to a permutation~$\widehat{g_i}$ of~$V(T)$ as follows.
The permutation~$\widehat{g_i}$ maps each vertex~$v$ to~$\varphi_{([v],g_i([v]))}(v)$.  Since~$g_i$ is a permutation and each~$\varphi_{(C,C')}$ is a bijection, the map~$\widehat{g_i}$ is a permutation of~$V(T)$. Let~$\widehat{\Psi}= \{ \widehat{g_1},\ldots,\widehat{g_t} \}$ be the set of lifted generators.

As next step, for each class~$C$ we recursively compute a generating set~$\Upsilon_C$ for~$\Aut(T[C])$. We lift each generator in~$\Upsilon_C$ to a permutation of~$V(T)$ by fixing all vertices outside of~$C$ obtaining the set~$\widehat{\Upsilon}_C$ of lifted generators.

Consider the group~$\Gamma$ generated by the set~$\widehat{\Psi} \cup \bigcup_{C\in \pi} \widehat{\Upsilon}_C$. 
 As a last step, using Theorem~\ref{thm:solvable:intersect} we compute the subgroup~$\Gamma' =\Gamma \cap \Aut(T)$.

The details of this algorithm are given in Algorithm~\ref{alg:aut:group}.

\medskip

\emph{(Running time)} 
We first argue that all work performed by an iteration of the algorithm apart from the recursive calls is polynomial in~$n$, say~$O(n^c)$ for some constant~$c$. This is obvious for all instructions of the algorithm except the task to compute the intersection of $\langle \Psi\rangle\cap \Aut(T)$ in Case 0 (Line~\ref{item:case:0:stab}) and the task to compute $\langle \widehat{\Psi}\rangle \cap \Aut(T)$  in Case 3b (Line~\ref{item:case:3b:stab}).
However, in Case 0, the group generated by $\langle \Psi\rangle$ is a direct product of solvable groups, thus solvable, and in Case 3b, the group $\langle \widehat{\Psi}\cup  \bigcup_{C\in \pi} \widehat{\Upsilon}_C\rangle$ is a subgroup of a wreath product of a solvable group with a solvable group and is thus solvable. (Alternatively we can observe that the natural homomorphism from the group $\langle \widehat{\Psi}\cup  \bigcup_{C\in \pi} \widehat{\Upsilon}_C\rangle$ to~${\Psi}$ has kernel~$\langle\bigcup_{C\in \pi} \widehat{\Upsilon}_C\rangle$, a direct product of solvable groups.) In either case, using the algorithm from Theorem~\ref{thm:invariant:suborbits:give:tour:iso}, the group intersection can be computed in polynomial time.

It remains to consider the number of recursive calls.
We will bound the amount of work of the algorithm in terms of~$t$, the maximum size of a color class of~$T$, and the number of vertices~$n$. Denote by~$R(t,n)$ the maximum number of nodes in the recursion tree over all tournaments for which the color classes have size at most~$t$ and  the number of vertices is at most~$n$. Note that~$R(t,n)$ is monotone increasing in both components. 

First note that if~$t<n$ the algorithm will end up in Case~0. 
The recursive bound in Case~0 (Line~\ref{item:case:0:recursion}) is then~$R(t,n) \leq 1+  \sum_{i = 1}^\ell R(a_i,a_i)$ for some positive integers~$a_1,\ldots,a_{\ell}\in \mathbb{N}$ (the color class sizes)  that sum up to~$n$ but are smaller than~$n$.

In Case~1, we have~$t=n$. The tournament~$T'$ is colored into three color classes. 
Since~$T$ is transitive (and thus every vertex has in- and out-degree~$(t-1)/2$), in~$T'$ there is one color class of size~$1$ and there are two classes of size~$(t-1)/2$.
The recursive call will lead to Case~0, which then
yields one trivial recursive call on a tournament of size~1 and two calls with tournaments of size~$(t-1)/2$.
We obtain a recursive bound (for Line~\ref{item:case:1:recursion}) of~$R(t,n)\leq  2+ 2 R((t-1)/2,(t-1)/2)\leq  3 R(t/2,t/2)$.

In Case 2, we have~$t=n$ and observe that the recursive call is for a tournament that is not monochromatic. Thus the recursive call will end up in Case 0. We thus obtain a recursive bound (for Line~\ref{item:case:2:recursion}) of~$R(t,n) \leq  2+ \sum_{i = 1}^\ell R(a_i,a_i)$ for some positive integers~$a_1,\ldots,a_{\ell}\in \mathbb{N}$ that sum up to~$n$ but are smaller than~$n$.

In Case 3, we have~$t=n$. Note that if the classes of~$\pi$ have size~$t'$ then~$t'\leq n/3$ (elements of~$\pi$ are all equally large and there are at least 2 but there is an odd number) and there are~$(n/t')^2=(t/t')^2$ recursive calls in Line~\ref{line:iso:call:via:aut}.
In the graph~$T_{C,C'}$ the color classes have size at most~$3 (t'-1)/2$ and there are at most~$3t'$ vertices. (The increase of a factor 3 comes from the~$\mathrm{Tri}()$ operation.) 
Thus the cost for such calls is bounded by~$ (t/t')^2 \cdot R(3(t'-1)/2,3t')\leq (t/t')^2 \cdot R(3t'/2,3t') \eqqcolon R_3 $, where~$3t'/2\leq n/2 = t/2$.

Using the same arguments as before, in Case 3a we thus get a recursive bound for Line~\ref{item:case:3a:recursion} of~$\sum_{i = 1}^\ell R(b_i,b_i)$ and thus for Case 3a in total a bound of~$R(t,n) \leq 1 + R_3+ \sum_{i = 1}^\ell R(b_i,b_i)$ for some positive integers~$b_1,\ldots,b_{\ell}\in \mathbb{N}$ that sum up to~$n$ but are smaller than~$n-t'= t-t'$.

In Case 3b we need to additionally consider the cost for the recursive call in Line~\ref{item:quotient}. This cost is at most~$ R(t/t',t/t')$ where~$t'\geq 2$ since the coloring is not discrete.
Also there is a recursive cost of~$t/t'\cdot  R(t',t') $ coming from Line~\ref{item:auts:of:parts}. 
Thus in this case we end up with~$R(t,n) \leq 1 + R(t/t',t/t') + R_3 + t/t'\cdot  R(t',t')$.

Summarizing we get that~$R(t,n)$  is bounded by
\[  \begin{cases}
1 & \text{if $n=1$}\\
2+  \sum\limits_{i = 1}^\ell R(a_i,a_i) & \text{in Cases 0 and 2, with~$\sum\limits_{i=1}^{\ell}a_i = n$ and $a_i\leq n-1$}\\
3 R(t/2,t/2)& \text{in Case 1}\\
1 + R_3+ \sum\limits_{i = 1}^\ell R(b_i,b_i)& \text{in Case 3a, with~$\sum\limits_{i=1}^{\ell}b_i = n$ and~$b_i\leq t-t'$}\\
1 + R(t/t',t/t') + R_3 + t/t'\cdot  R(t',t') & \text{in Case 3b,}
\end{cases}\]

where~$ R_3 = (t/t')^2 \cdot R(3t'/2,3t')$ and~$t'$ satisfies~$3t'/2 \leq t/2$ and~$t/ t'\leq t/2$ and~$3t'\leq t$.
Let us define~$S(m)$ as the maximum of~$R(t,n)$ over all pairs of positive integers~$(t,n)$ with~$t+n\leq m$ and~$t\leq n$. Then we get from the above considerations that~$S(m)$ is bounded by one of the following

\[  \begin{cases}
1 & \text{if $m=2$}\\
2+  \sum\limits_{i = 1}^\ell S(a_i) & \text{$\sum\limits_{i=1}^{\ell}a_i \leq m$ and $a_i\leq m-1$}\\
3 S(m/2)& \\
1 + (m/t')^2 \cdot S(9/2 \cdot t')+ \sum\limits_{i = 1}^\ell S(b_i)& \sum\limits_{i=1}^{\ell}b_i = m$ and~$b_i\leq m-t'\\
1 + S(m/t') + (m/t')^2 \cdot S(9/2 \cdot t')+ m/t'\cdot  S(t'), & 
\end{cases}\]
where~$t'$ satisfies~$9/2 \cdot t'\leq 3/4\cdot m$ and~$t/ t'\leq m/2$.
It is now simply a calculation to show that for~$d$ sufficiently large, the function~$F(m)= m^d$ satisfies all the recurrence bounds for~$S$ (of course as lower bounds rather than upper bounds).
We show the calculation for the most interesting case, the Case 3a. 
Let~$x = m-t'$. Then~$5t'\leq x$. Furthermore the equation for Case 3a says~$S(x+t')\leq 1+ ((x+t')/t')^2 S(9/2 t') + \sum\limits_{i = 1}^\ell S(b_i)$ where~$b_i\leq x$ and~$\sum\limits_{i=1}^{\ell}b_i = x+t'$. Note for the function~$F$ that~$\sum\limits_{i = 1}^\ell F(b_i)$, under the conditions~$b_i\leq x$ and~$\sum\limits_{i=1}^{\ell}b_i = x+t'$, gets maximized as~$x^d + (t')^d$.
For the right hand side we get~$1+ ((x+t')/t')^2  (9/2)^d (t')^d + x^d + (t')^d \leq 1+ x^d+ (6/5)^2 (9/2)^d x^2  (t')^{d-2} +  (t')^d$ which is certainly bounded by~$(x+t')^d$ for~$d$ sufficiently large since the expansion of~$(x+t')^d$ contains the summands~$x^d$,~$(t')^d$ and~$d x^{d-1} t'\leq d 5^{d-3} x^{2} (t')^{d-2}$. Thus~$F$ is an upper bound for~$S$.

Overall we obtain a polynomial-time algorithm from this recursive bound. This in particular implies that the algorithm halts.

\medskip

\emph{(Correctness)} 
For the correctness proof we analyze the different cases one by one. By induction we can assume that recursive calls yields correct answers.

For Case~0, since the last instruction intersects some group with the automorphism group it is clear that the algorithm can only return automorphisms of~$T$. Let us thus assume that~$\varphi\in \Aut(T)$. Then, for each color class~$c$, the set~$V_c$ is invariant under~$\varphi$ and~$\varphi|_{V_c}\in  \Aut(T[V^c])$. This implies that~$\varphi \in \langle \Psi\rangle$.

For Case 1,~$T$ is transitive since~$\pi = V(T)$ shows that~$V(T)$ is an orbit. Thus,~$\Aut(T)$ is generated by the point stabilizer~$\Aut(T)_v \coloneqq \{\psi \in \Aut(T)\mid \psi(v) = v\}$  and an arbitrary transversal (i.e., a subset of elements of~$\Aut$ containing a representative from each coset of~$\Aut(T)_v$ in~$\Aut(T)$). Since~$\Phi$ contains a certificate for all pairs of distinct vertices~$(v,v')$ and since~$\Phi\subseteq \Aut(T)$ we conclude that~$\Aut(T) = \langle \Phi\cup \Aut(V(T'))\rangle$.

For Case 2, it suffices to note that for every integer~$i\in \mathbb{N}$ the set $\{v \in V(T) \mid |[v]_\pi| = i \}$ is invariant under~$\Aut(T)$. 

For Case 3, Line~\ref{line:iso:call:via:aut} note that similar to Case 1, the graphs~$T[C]$ and~$T[C']$ are transitive and thus the individualization in Line~\ref{item:individ:tri} does not make isomorphic graphs non-isomorphic.

For Case 3a, again note that for~$v\in V(T)$ the set~$\{v' \in V(T) \mid T[([v])] \cong T[([v'])] \}$ is invariant. For Case 3b we argue similarly to Case~0. Since the last  instruction intersects some group with the automorphism group it is clear that the algorithm can only return automorphisms of~$T$. Let us thus assume that~$\varphi\in \Aut(T)$.

Then~$\varphi$ induces an automorphism~$\psi$ of~$T/\pi$. Note that there is some~$\widehat{\psi}$ in~$\langle\widehat{\Psi}\rangle$ that also induces~$\psi$ on~$T/\pi$. It suffices now to show that the map~$\widehat{\psi}^{-1}\circ \varphi$ is in~$\langle\bigcup_{C\in \pi} \widehat{\Upsilon}_C\rangle$. Consider~$C\in \pi$. Then~$\widehat{\psi}^{-1}\circ \varphi$ maps~$C$ to~$C$ and more strongly it induces an automorphism of~$T[C]$ which must be contained in~$\langle\widehat{\Upsilon}_C\rangle$. We conclude that~$\widehat{\psi}^{-1}\circ \varphi$ is in~$\langle\bigcup_{C\in \pi} \widehat{\Upsilon}_C\rangle$ finishing the proof.
 \end{proof}
We have now assembled all the required parts to prove the main theorem of the paper.

\begin{corollary}
\begin{enumerate}
\item There is a randomized (one-sided error) polynomial-time Turing reduction from tournament isomorphism to asymmetry testing of tournaments (i.e.,~$\GIlang_\mathcal{\Tourlang} \leq_{r,T}^p \GAlang_\mathcal{\Tourlang}$).
\item There is a randomized polynomial-time Turing reduction from the computational task to compute generators of the automorphism group of a tournament to asymmetry testing of tournaments (i.e.,~$\AUTlang_\mathcal{\Tourlang} \leq_{r,T}^p \GAlang_\mathcal{\Tourlang}$).
\end{enumerate}

\begin{proof}
Recall that a two-sided error algorithm for an isomorphism search problem can be readily turned into a one-sided error algorithm by checking the output isomorphism for correctness. Thus, by Lemma~\ref{lem:asym:iso:red:to:asym} Part~\ref{item:gi:to:aut}
it suffices to prove the second part of the corollary.

Combining Lemma~\ref{lem:from:tour:asym:to:aut:sampler} and Theorem~\ref{thm:from:aut:sampler:to:suborbits}, from an oracle to tournament asymmetry we obtain a randomized Monte Carlo (i.e., with possible errors) algorithm that computes invariant suborbits. Given a Las Vegas algorithm (i.e., no errors) for suborbits, the previous theorem provides us with a computation of the automorphism group of tournaments. 

It remains to discuss the error probability we get from using a Monte Carlo algorithm instead of a Las Vegas algorithm. Since there is only a polynomial number of oracle calls, and since the error bound in Theorem~\ref{thm:from:aut:sampler:to:suborbits} can be chosen smaller than~$\frac{1}{|G|^c}$ for every fixed constant~$c$, the overall error can be chosen to be arbitrarily small.
\end{proof}

\end{corollary}
\section{Discussion and open problems}\label{sec:open:prob}

This paper is concerned with the relationship between the asymmetry problem~$\GAlang_\mathcal{\mathcal{C}}$ and isomorphism problem~$\GIlang_\mathcal{\mathcal{C}}$. While under mild assumptions there is a reduction from the former to the latter, a reduction in the other direction is usually not known. However, for tournaments we now have such a randomized reduction.

The first question that comes to mind is whether the technique described in this paper applies to other graph classes. While the sampling techniques from Sections~\ref{sec:sampling:subsets} to~\ref{sec:sampling:minimimal:orbits} can be applied to all graph classes that satisfy mild assumptions (e.g.,~$\cGIlang_\mathcal{\mathcal{C}} \leq_t^p \GIlang_\mathcal{\mathcal{C}}$ and~$\cGIlang_\mathcal{\Asymlang\mathcal{C}} \leq_t^p \GIlang_\mathcal{\Asymlang\mathcal{C}}$) the algorithm described in Section~\ref{sec:auto:group:from:suborbits} crucially uses the fact that automorphism groups of tournaments are solvable. This is not the case for general graphs, so for the open question of whether~$\GIlang$ reduces to~$\GAlang$ this may dampen our enthusiasm. However, what may bring our enthusiasm back up is that there are key classes of combinatorial objects that share properties similar to what we need.

In particular, this brings us to the question whether the techniques of the paper can be applied to group isomorphism. Just like for tournament isomorphism, finding a faster algorithm for group isomorphism (given by multiplication table) is a bottleneck for improving the run-time bound for isomorphism of general graphs beyond quasi-polynomial.
Since outer-automorphism groups of simple groups are solvable, we ask: Can we reduce the group isomorphism problem to the isomorphism problem for asymmetric groups? This question is significant since an asymmetry assumption on groups is typically a strong structural property and may help to solve the entire group isomorphism problem.
However, here one has to be careful to find the right notion of asymmetry since all groups have inner automorphisms. For such notions different possibilities come to mind. 

A second natural open question would be whether there is a deterministic version of the algorithms given in this paper. 

As a last open problem recall that it was shown in Section~\ref{sec:sampling:subsets}  that one can extract a characteristic subset for a sampler over a set~$M$ in time that depends polynomially on~$M$. Since the automorphism group of a graph can be superpolynomial in the size of the graph, we had to take a detour via suborbits in Section~\ref{sec:sampling:minimimal:orbits}. There can be no general way to extract a characteristic subset of~$M$ in polynomial time if~$|M|$ is not polynomially bounded, since we might never see an element twice.
However, if~$M$ has an algebraic structure, in particular if~$M$ is a permutation group over a polynomial size set, this is not clear. 
Thus we ask: Is there a polynomial-time (randomized) algorithm that extracts a characteristic subgroup using a sampler~$\Gamma$ over a permutation group?

\bibliographystyle{plainurl}
\bibliography{main}

\end{document}